\documentclass[review]{elsarticle}

\usepackage{hyperref}
\usepackage{amssymb,amsmath}
\usepackage{graphicx}
\usepackage{color}
\newtheorem{proposition}{\textbf{Proposition}}
\newtheorem{theorem}{\textbf{Theorem}}
\newtheorem{proof}{\textbf{Proof}}

\newtheorem{remark}{\textbf{Remark}}
\usepackage[ruled,vlined]{algorithm2e}
\usepackage{caption, subcaption}
\journal{International Journal of Critical Infrastructure Protection}

\begin{document}

\begin{frontmatter}

\title{A Dimensionality-Reduction Strategy to Compute Shortest Paths in Urban Water Networks}

\author[1]{Carlo Giudicianni\corref{cor1}} 
\cortext[cor1]{Corresponding author: 
  Tel.: +39-081-000-0000;}
\ead{carlo.giudicianni@unicampania.it}
\author[1,2]{Armando Di Nardo}
\author[3]{Gabriele Oliva}
\author[2]{Antonio Scala}
\author[4]{Manuel Herrera}

\address[1]{Dipartimento di Ingegneria, Universit\`a degli Studi della Campania 'L. Vanvitelli', via Roma 29, Aversa 81031, Italy}
\address[2]{Institute for Complex Systems -- Italian National Research Council, via dei Taurini 19, Roma 00185, Italy}
\address[3]{Dipartimento di Ingegneria, Universit\`a Campus Bio-Medico di Roma, via Álvaro del Portillo 21, Roma 00128, Italy}
\address[4]{Institute for Manufacturing -- Dept. of Engineering, University of Cambridge, 17 Charles Babbage Rd., CB3 0FS Cambridge, United Kingdom}

\begin{abstract}
The efficient computation of shortest paths in complex networks is essential to face new challenges related to critical infrastructures such as a near real-time monitoring and control and the management of big size systems. In particular, using information on the minimum paths in water distribution networks (WDNs) allows to track the diffusion of contaminants and to quantify the resilience and criticality of the system. This is, ultimately, approached by considering dynamically changing path-weights that depend on the flow or on other information available at run-time. These analyses tipically include all the WDN assets but reducing the high degree of physical details with a minimum lost of key information for their performance assessment. This paper proposes a strategy to compute minimum paths that is based on a dimensionality-reduction process. Specifically, the network is partitioned into communities and suitably modified to obtain a reduced complexity representation (e.g., in terms of number of nodes and links). The paper shows how this novel, reduced representation is equivalent to the traditional network on computing the shortest paths. The proposed approach is validated considering two utility networks as case studies. The results show that the proposed method provides the exact solution for the shortest path with a computational-time reduction consistently over 50\% and up to 90\% for some cases. Furthermore, the application of the proposal on WDNs partitioning shows both hydraulic and economic advantages thanks to their monitoring and controlling at near real-time.
\end{abstract}


\end{frontmatter}


\section{Introduction}
\label{sec1}
Modern society is strongly dependent on infrastructure systems (i.e. transportation, power grids, telecommunications, water systems), which support city growth and economic prosperity. These infrastructures continually face natural and man-made threats that cause economic and social disruption, leading their operators to continuously work on improving safety and security and on speeding up mitigation actions. Today, the reliability and performance assessment, continuous operation, monitoring, and protection of critical infrastructures are national priorities for countries worldwide \citep{alcaraz2015critical}. Furthermore, as cities increase in their size, these infrastructures are getting larger and tangled, showing a complex behaviour due to the high degree of interdependency among them. As a consequence, the management of such infrastructures is becoming an arduous task to address, and there is a need to develop new, agile tools and methodologies to analyse and control them. This is especially true for water distribution networks (WDNs), which are among the most important civil networks as they deliver drinking and industrial water to metropolitan areas. WDNs constitute essential and critical infrastructures, as systems whose operability is of crucial importance to ensure social survival and welfare. They must to face two major threats:
\begin{itemize}
    \item $Contamination$: water infrastructures are strongly vulnerable to malicious and intentional attacks since they are made up of thousands of  exposed  elements \citep{chianese2017dma}; as a consequence, water can be easily polluted by chemical or biological contaminants, which spread all over the system by flowing, with dramatic effects on the citizen health \citep{kroll2006securing};
    \item $Leakages$: water infrastructures are constituted by aged buried pipelines, and as a consequence, pipelines are easily eroded by the moist environment. Furthermore, the daily pressure variability strongly stresses the pipes. These factors lead to the failure and burst of the pipes \citep{wang2019systemic}, causing leakages and wasting huge amount of water \citep{us2010control}.
\end{itemize}

The downside in the management of such infrastructures is that the underlying details of the physics involved in their functioning complicates the analysis to a relevant extent, making it difficult to achieve useful insights in reasonable time \citep{krause2008efficient}. Complexity Science has proven to be a particularly adequate tool for the timely and agile analysis and management of WDNs (and more in general for critical infrastructures) \citep{stergiopoulos2015risk}, especially in the case of limited information about the system \citep{torres2016exploring, Giudicianni2018}. In particular, a complex network representation of infrastructures allows to abstract away from the high degree of physical details of the systems and focus only on a few crucial aspects in a manageable way \cite{DiNardo2018Ap,wang2010generating,beyza2019applying,haznagy2015complex}. At this end, knowledge on the minimum paths (eventually considering dynamically changing weights that depend on the flow or on other information available at run-time) allows to track the diffusion of contaminants and quantify the resilience and criticality of the system and its composing elements \cite{di2017simplified,tinelli2017sampling}.

\textcolor{black}{The work of \cite{Fu2006} encompasses an extensive survey of various heuristic {\em shortest path} (SP) algorithms developed in the last years. It is worth to mention the interesting strategy adopted for practitioners and applied researchers to exploit network's domain-specific information. This is the case of traffic systems researchers adopting the natural hierarchies of the roads to signiﬁcantly speed up the SP computational time \cite{Jagadeesh2002, Jung2002}. Overall, there are two widely investigated strategies for approximate the SP computation in large-scale complex networks. One of them is the $landmark-based$ method. This requires to pre-compute the shortest paths between special nodes (landmark nodes) and all the other nodes in the network, saving these distances in a database. The shortest-path between two nodes is, then, approximated by combining those distances stored in the database \cite{Tang2003,kleinberg2004triangulation, Vieira2007}. The other one is a $topology-based$ approach. This strategy lies on the structure of networks through their partition into discrete areas \cite{Djidjev1996, Henzinger1997, Jing1998, Gutman2004}. In this regard, \cite{Potamias2009} propose an approximated landmark-based method for point-to-point distance estimation in large-scale networks, also adding the partitioning variant. The landmark set is selected for each network area and the shortest paths consequently saved in a database. The Authors also demonstrate that selecting the optimal set of landmark nodes is an \textit{NP-hard} problem.}

This paper proposes a strategy to compute minimum paths, namely {\em Multiscale Shortest Path} (MS-SP), that is based on a dimensionality-reduction process. Specifically, we make a network partition into communities and suitably modify such network to obtain a representation with reduced complexity (e.g., in terms of the number of nodes and edges), namely {\em Multiscale Network} (MS), but equivalent in terms of computing the shortest paths. The proposed approach is validated considering two real-world WDNs as case studies.

An antecedent of this paper can be found on the work of \cite{gong2016efficient} which provides a valid approximation to the shortest path problem for social networks. In such work, the authors propose a combined process for community detection and network reduction, by collapsing communities into nodes of a new network. The proposed algorithm has a similar scope, but the main innovation proposed herein is that the network reduction process identifies a subset of key nodes that lie at the boundary of the communities and transforms the  community into hyper-links connecting such boundary nodes, rather than collapsing the communities in single nodes. As a result, the network collapses into a reduced-size graph were boundary nodes are interconnected by edges that are weighted in a suitable manner to guarantee that the minimum path between two nodes in the original network can be computed in terms of the minimum path between the boundary nodes that are closest to the source and destination, respectively.

Two utility networks validate the proposal. One is the well-known network of Colorado Springs \cite{Lippai2005} at Colorado, US. This network is investigated to facilitate repeating further the proposed MS-SP algorithm. The second case-study corresponds to the operational network supplying water to Alcal\'a de Henares (Madrid, Spain). In this network is more evident the usefulness of the proposal to speed-up the SP computation in large-scale, urban infrastructure networks. In addition to the advantages of the compuation, the dimensionality-reduction process leads to a novel representation of WDNs where it is even possible to obtain a visualisation of the shortest paths.

The outline of the paper is as follows: in Section \ref{sec:methods} the proposed algorithm is presented from a mathematical point of view. Section \ref{sec:experimental} shows the two case studies. Section \ref{sec:results} reports the simulation results with special emphasis on their associated computational efforts. Section \ref{sec:discussion} presents a further application of the algorithm for the management of water networks. The paper closes with a conclusions section which also points out future research directions.

\section{Methods}
\label{sec:methods}

\subsection{Graphs}
Let's $G=\{V,E,W\}$ denote a {\em weighted graph} with a finite number $n$ of nodes $v_i\in V$ with $i\in\{1,\ldots, n\}$ and edges \mbox{$(v_i,v_j)\in E\subset V\times V$} from node $v_i$ to node $v_j$.
For each edge $(v_i,v_j)\in E$ we denote by $w_{ij}\in W$ the associated weight.
A graph is said to be {\it undirected} if $(v_i,v_j)\in {E}$ whenever $(v_j,v_i)\in {E}$, and it is said to be {\it directed} otherwise. In the following we will consider undirected graphs.
For undirected graphs, we assume the weights satisfy $w_{ij}=w_{ji}$ for all $(v_i,v_j)\in {E}$.
Let's the {\em weighted adjacency matrix} of a graph $G=\{V,E,W\}$ be the $n \times n$ matrix $A$ with the same structure as $G$, i.e., such that $A_{ij}=w_{ij}$ if  $(v_i,v_j)\in {E}$ and $A_{ij}=0$, otherwise.
In the case of undirected graphs, matrix $A$ is symmetric.
A {\em path} over a graph $G=\{V,E,W\}$, starting from a node $v_i\in V$ and ending in a node $v_j\in V$, is a subset of links in $E$ that connect $v_i$ and $v_j$; the {\em length} of the path is the sum of the weights associated to the links in the path.
A {\em minimum path} that connects $v_i$ and $v_j$ is the path from $v_i$ to $v_j$ of minimum length.
An undirected graph is {\em connected} if for each pair of nodes $v_i,v_j\in V$ there is a path over $G$ that connects them.

\subsection{Shortest path algorithm}
One of the most well known algorithms to compute the shortest path between two nodes in a weighted graph is Dijkstra's Algorithm (D-SP) \cite{Dijkstra1959}, which can be summarized as follows. Given a weighted graph $G=\{V,E,W\}$ with $|V|=n$ nodes, a start node $v_s$ and a goal node $v_g$, the algorithm keeps track of three variables for each node: 
\begin{itemize}
\item $\texttt{visited}(v_i)$ which is equal to one if the node has already been visited during the algorithm and is zero otherwise;
\item $\texttt{distance}(v_i)$ which is the current estimate for the distance of node $v_i$ from the start node $v_s$;
\item $\texttt{parent}(v_i)$ which is the identifier of the node immediately before node $v_i$ in the path connecting $v_s$ and $v_i$.
\end{itemize}
Moreover, the algorithm keeps track of the node currently being examined, which is referred to as $v_*$.

During the initialisation phase, the algorithm sets $\texttt{visited}(v_s)=1$ and $\texttt{visited}(v_i)=0$, for all $v_i\in V\setminus\{v_s\}$. Moreover, it sets $\texttt{distance}(v_s)=0$ and $\texttt{distance}(v_i)=\infty$, for all $v_i\in V\setminus\{v_s\}$. Finally, the algorithm selects $\texttt{parent}(v_i)=\emptyset$ for all $v_i\in V$ and  sets $v_*=v_s$.
Then, the main cycle of the algorithm is executed; such a main cycle is composed of the following conceptual steps:
\begin{itemize}
\item [Step 1] For all neighbours $v_i$ of $v_*$ such that $\texttt{visited}(v_i)=0$ set the distance of node $v_i$ from $v_s$ as the minimum between the previous estimate and the sum of the distance of $v_*$ from $v_s$ and the weight of the link $w_{*\,i}$ connecting $v_*$ and $v_i$, i.e., 
$$
\texttt{distance}(v_i)=min\left\{\texttt{distance}(v_i),\texttt{distance}(v_*)+w_{*\,i}\right\};
$$
moreover, if the distance is updated for node $v_i$ the algorithm keeps track of the fact that the minimum path from $v_s$ to $v_i$ features the edge $(v_*,v_i)$ by setting
$$
\texttt{parent}(v_i)=v_*.
$$
\item [Step 2] Set $\texttt{visited}(v_*)=1$ 
\item [Step 3] If $\texttt{visited}(v_t)=1$ then stop, the algorithm is terminated.
\item [Step 4] Otherwise, select the node with minimum current distance among the not visited ones as the new current node, i.e.,
$$
v_*=v_j,\quad \mbox{ where } j=\underset{i\,| \texttt{visited}(v_i)=0}{\arg\min} \{\texttt{distance}(v_i)\}
$$
and go back to Step 3.
\end{itemize}

Notice that a straightforward application of the above algorithm yields a computational complexity $\mathcal{O}(|V|^2)$ where $|V|$ is the number of nodes in the graph; moreover, when the graph is particularly sparse, i.e., when $|E|\ll |V|(|V|-1)/2$, where $|E|$ is the number of edges,  it is possible to reduce complexity by using an implementation that relies on data structures such as the so-called Fibonacci heaps  \citep{fredman1987fibonacci}.

\subsection{Clustering approach}
A network community detection algorithm \cite{fortunato2016community} can be used in case the initial partition of the network is not available. These communities (clusters) are formed by grouping elements with similar characteristics or with a higher connection density than that external to the community. There is a large set of community detection algorithms proposed in literature. 
\textcolor{black}{In this paper the Louvain method \cite{blondel2008fast} is adopted, which uses an iterative process to improve the scalability of the overall community detection. Louvain algorithm is a heuristic method based on modularity optimisation \cite{newman2004fast}. In particular, it is divided in two iteratively repeated phases.} 

\textcolor{black}{\begin{itemize}
    \item The algorithm starts assigning a different community to each node of the network, then, for each node $i$ the neighbour $j$ is considered and the gain of modularity that would take place by removing $i$ from its community and by placing it in the community of $j$ is evaluated. After that, the node $i$ is placed in the community for which this gain is maximum. The process is applied for all nodes until no further improvement can be achieved. The first phase stops when a local maxima of the modularity is attained, and no individual move can improve the modularity.
    \item The second phase of the algorithm builds a new network whose nodes are the communities found during the first phase; once it is completed, the first phase of the algorithm is reapplied to the resulting network. The passes are iterated until the maximum of modularity is attained.
\end{itemize}}

It is known that, the Louvain algorithm appears to run in time $\mathcal{O}(|E|)$, where $|E|$ is the number of edges in the graph \citep{traag2015faster}.

\subsection{Multiscale Shortest Path (MS-SP) algorithm}
Given a graph $G=\{V,E,W\}$, the proposed approach to calculate the shortest path from a node $v_s$ to a node $v_t$ is based on a dimensionality reduction procedure, where the network is decomposed into clusters and the nodes/edges in each clusters are collapsed in a suitable way that guarantees that the shortest path computed over the resulting graph corresponds to the one on the original graph.

Specifically, we apply the Louvain clustering algorithm to $G$, decomposing the set of nodes $V$ into $q$ disjoint sets $V_1,\ldots,V_q$, each corresponding to a cluster.
In the following, we denote by $E_i$ the set of edges in the original edge set $E$ that connect nodes in the same cluster, i.e.,
$$
E_i=\{(v_a,v_b)\in E\,|\, v_a,v_b\in V_i\};
$$
moreover, we define
$$
\hat E_{ij}=\{(v_a,v_b)\in E\,| v_a\in V_i \mbox{ and } v_b\in V_j\}
$$
and 
$$E_{ij}=\hat E_{ij}\bigcup\hat E_{ji}.$$
Finally, we define the set of {\em boundary nodes} $V_i^b\subseteq V_i$ as the set of nodes in $V_i$ that belong to at least one edge in $E_{ij}$ for some $j\in\{1,\ldots,q\}\setminus \{i\}$, i.e.
$$
V_i^b=\{v_a\in V_i\,|\,\, \exists (v_a,v_b) \in E,\,\, v_b\not\in V_i\}.
$$
 In other words, $E_{ij}$ is the set of edges that connect nodes in $V_i$ and nodes in $V_j$, and it holds $E_{ij}=E_{ji}$.
Specifically, by running the clustering procedure described above, the network is decomposed into $q$ clusters.

The dimensionality reduction strategy consists in the construction of a graph
$$
\widetilde G=\{\widetilde V,\widetilde E,\widetilde W\},
$$
where $\widetilde V$ includes the set of boundary nodes and the start and goal nodes, i.e.,
$$\widetilde V=\{v_s,v_t\}\bigcup_{i=1}^q V_i^b.$$
As for the edge set $\widetilde E$, we have that
$$\widetilde E=\widetilde E_{\texttt{in}}\bigcup \widetilde E_{\texttt{out}},$$
where $\widetilde E_{\texttt{out}}$ is the union of the edges connecting boundary nodes, i.e.,
$$
\widetilde E_{\texttt{out}}=\bigcup_{i,j\in \{1,\ldots,q\}} E_{ij}
$$
and $\widetilde E_{\texttt{in}}$ is the union of sets $\widetilde E_{\texttt{in}}^i$ of edges that directly connect the boundary nodes in the $i$-th cluster. Note that, if the start or goal nodes are in the $i$-th cluster, then the start or goal nodes are considered as a boundary node.

As for the weights, we select $\widetilde w_{ab}= w_{ab}$ whenever $(v_a,v_b)\in \widetilde E_{\texttt{out}}$, while for each  pair of boundary nodes $v_a,v_b$ that belong to the same cluster $i$ (including the start or goal node if they belong to cluster $i$), we compute the minimum path $p^i_{ab}$ between $v_a$ and $v_b$ over the subgraph of $G$ induced by considering just the nodes $V_i$ in the $i$-th cluster and we set the weight as the length of the path $p^i_{ab}$, i.e.,
$$
w_{ab}=\sum_{(v_h,v_k)\in p^i_{ab} }w_{hk} .
$$
At this point, the algorithm finds the minimum path between nodes $v_s$ and $v_t$ by computing  the minimum path between $v_s$ and $v_t$ over $\widetilde G$. Note that, by keeping track of the minimum paths involving boundary nodes in each cluster (treating $v_s$ and $v_t$ as boundary nodes), we are able to reconstruct the minimum path over $G$ in terms of the minimum path over $\widetilde G$.
 
The algorithm is graphically explained by Figure~\ref{fig:MS-SP} in which there are 3 groups: the upper-left cluster contains three boundary nodes (and the start node), the right cluster has three boundary nodes and the lower cluster has two boundary nodes (plus the target node). As a result of the decomposition, we obtain a network with $|\widetilde V|=10$ nodes (i.e., the boundary nodes plus the start and goal) and $|\widetilde E|=16$ edges; in particular, the four edges connecting nodes in different clusters are kept, while for each pair of boundary (or start/goal) nodes in each cluster a new link is added, whose weight corresponds to the length of the minimum path, computed over the subgraph induced by the nodes in the cluster. At this point, the minimum path is computed by computing the minimum path over $\widetilde G$.

\begin{figure}[!hbt]
    \centering
    \begin{subfigure}[b]{0.4\linewidth} 
        \centering
        \includegraphics[width=\textwidth]{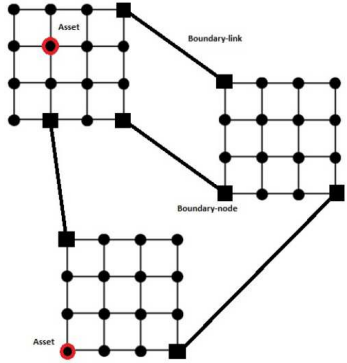}
        \caption{Original network showing key elements}
        \label{fig:MS-SP1}
    \end{subfigure}
    \hfill
    \begin{subfigure}[b]{0.4\linewidth}
        \centering
        \includegraphics[width=\textwidth]{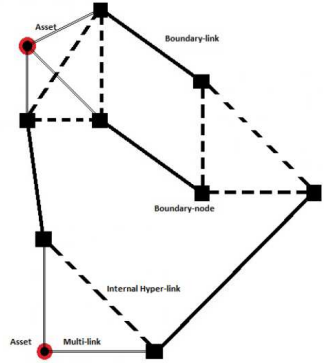}
        \caption{MS network showing key elements} 
        \label{fig:MS-SP2}
    \end{subfigure}
    \caption{Graphical explanation of the dimensionality reduction process for computing the SP algorithm}
    \label{fig:MS-SP} 
\end{figure}

\subsection{Correctness of the algorithm}
The following theorem establishes that the path found via the proposed algorithm is, indeed, a minimum path.
\begin{theorem}
The minimum path between nodes $v_s$ and $v_t$ over $\widetilde G$ is equivalent to the one connecting $v_s$ and $v_t$ over the original graph $G$.
\end{theorem}
\begin{proof}
Let's $v_s$ and $v_t$ belonging to the same cluster in $G$. Then, by construction, the minimum path found over $\widetilde G$ corresponds to the one over $G$.
Let's assume now that $v_s$ and $v_t$ belong to different clusters with node sets $V_s$ and $V_t$. By construction, since the clusters are connected only via edges joining boundary nodes belonging to different clusters, the minimum path joining $v_s$ and $v_t$ in $G$ features a path from $v_s$ to a node $v_{s'}\in V_s$, a path from $v_{s'}$ to a node $v_{t'}\in V_t$ and path from a node $v_{t'}$ to $v_t$ (note that if $v_s=v_{s'}$ or $v_t=v_{t'}$ the path joining them is the empty set).

At this point, we observe that the path connecting $v_s$ to any $v_{s'}\in V_s$ and the path connecting $v_{t}$ to any $v_{t'}\in V_t$ are, by construction, minimum paths; similarly, the path connecting any $v_{s'}\in V_s$ and $v_{t'}\in V_t$ with (recall that we assumed $s\neq t$) is a minimum path. Hence, by construction, the minimum path found over $\widetilde G$ corresponds to a minimum path
 $p_{st}=p_{ss'}\bigcup p_{s't'}\bigcup p_{t't}$
 over the original graph, for some $v_{s'}\in V_s$ and $v_{t'}\in V_t$.
The proof is complete.
\end{proof}

\subsection{Time complexity of the MS-SP algorithm}
\textcolor{black}{In the following, it is shown the computational cost of the proposed algorithm. It is important to point out that, the core idea of working on a size-reduced graph does not depending on the chosen clustering algorithm. As a consequence, a faster method can be adopted making the proposed algorithm even more convenient from a computational point of view.} 

\begin{proposition}
The computational complexity of the proposed approach, including the clustering procedure and the construction of the reduced graph $\widetilde G$, is equal to
$$
max\left\{\mathcal{O}(|E|), \mathcal{O}(n_b^2), \mathcal{O}\left(\sum_{i=1}^q |V_i|^2 |V_i^b|^2\right)\right\},
$$
where $V_i$ is the set of nodes in the $i$-th cluster and $V_i^b$ is the set of boundary nodes in the $i$-th cluster and $n_b=\sum_{i=1}^q |V_i^b|$ is the cardinality of the set of all boundary nodes identified by applying Louvain algorithm.
\end{proposition}
\begin{proof}
The computational complexity of the Louvain method is $\mathcal{O}(|E|)$.
Moreover, once the clusters are formed, we need to scan all the edges to identify the set of boundary nodes, a procedure that requires $\mathcal{O}(|E|)$.

At this point, the proposed algorithm computes the shortest path over the subgraph induced by each cluster among each pair of boundary nodes in that cluster; each cluster has $|V_i|$ nodes, hence the computation of the shortest path from one node in the cluster to all other nodes in the cluster requires $O(|V_i|^2)$, since each cluster has $\mathcal{O}( |V_i^b|)$ distinct pairs of boundary nodes, we have that the computational complexity for each cluster is $\mathcal{O}(|V_i|^2 |V_i^b|^2)$.
Since the clusters are $q$ we get $\mathcal{O}(\sum_{i=1}^q |V_i|^2 |V_i^b|^2)$.

To conclude, the application of Dijsktra's algorithm on the reduced-size network has a complexity $\mathcal{O}(n_b^2)$; the proof follows since the two operations are done in series, hence the computational complexity is
equal to the largest among the computational complexities of the above procedures.
\end{proof}

Note that the computational complexity of the computation of the minimum path, after the graph  $\widetilde G$ has been created is remarkably smaller, being $n_b\ll|V|$ for real world networks. Similarly, the complexity of the clustering procedure, although being theoretically upper bounded by $\mathcal{O}(|V|^2)$, is likely to be remarkably smaller,  especially when the graph is sparse and $|E|\ll|V|(|V|-1)/2$, $|V|(|V|-1)/2$ being the number of edges in a complete graph.

As for the calculation of the minimum paths among the boundary nodes in the same cluster, we observe that there may be instances where complexity is above Dijsktra's algorithm\footnote{Consider for instance the case where the graph is full and is arbitrarily divided into $4$ clusters with the same number of nodes; in this extreme case $V_i^b=V_i$ and thus the complexity of the proposed algorithm would be $\mathcal{O}(|V|^4)$.}; however the likelihood of facing such instances is nearly zero in the case of WDNs and, in general, for graphs that have high sparsity and modularity. In fact, as discussed in the next remark, for those graphs the complexity of the construction of $\widetilde G$ is likely to be well below the one of Dijsktra's Algorithm. This fact is experimentally demonstrated in the next section.
\begin{remark}
Note that the complexity of computing the minimum paths locally at every cluster has a complexity
$
\mathcal{O}\left(\sum_{i=1}^q |V_i|^2 |V_i^b|^2\right)$.
However, especially when the network has a clear modular structure, the number $q$ of clusters is likely to be sublinear\footnote{ For instance, in \citep{Giudicianni2018} it is shown that for real WDNs the optimal number of clusters is $q\approx n^{0.3}$.} in $|V|$ (e.g., $q=|V|^{\gamma}$ with $\gamma\in(0,1)$). Hence, on average, also the cardinality $|V_i|$ of the node set of the clusters  is likely to be sublinear, i.e.,
$$
|V_i|\approx n/q= |V|^{1-\gamma}.
$$. 
Moreover, the cardinality of $V_i^b$ is likely to satisfy $|V_i^b|\ll |V_i|$ and, in several practical cases, can be assumed to be constant for planar graphs and WDNs (see \citep{Giudicianni2018}), i.e., $|V_i^b|\approx \mathcal{O}(1)$.
Hence, in practical cases of interest for this paper, we have 
$$
\mathcal{O}\left(\sum_{i=1}^q |V_i|^2 |V_i^b|^2\right)\approx
\mathcal{O}\left( |V|^{1+\gamma} \right)<\mathcal{O}\left( |V|^2\right).
$$
\end{remark}

\begin{remark}
Note that the construction of $\widetilde G$ can be slightly modified in order to be the base for the calculation of {\em all} shortest paths. In fact, it is sufficient to compute all shortest paths among every node in each cluster (i.e., requiring a computational complexity $\mathcal{O}(\sum_{i=1}^q |V_i|^2 |V_i|^2)=\mathcal{O}(\sum_{i=1}^q |V_i|^4)$) and storing information on the paths within each cluster. In this way, the graph $\widetilde G$ for calculating a path from any node $v_s$ to any node $v_t$ can be constructed by considering the links connecting boundary nodes and those connecting the start and goal nodes to the boundary nodes, an operation that requires at most $\mathcal{O}(|V|)$ in the worst case).
\end{remark}

\section{Experimental study}
\label{sec:experimental}
Urban utilities such as water, gas, or electric power networks can be modelled as quasi-planar graphs (e.g., networks forming vertices wherever two edges cross) with spatially organised weighted graphs $G=(V,E,w)$. In the case of water distribution systems the set $V$ of $n$ vertices/nodes encompasses junctions, water sources and demand points. The set $E$ of $m$ edges/links includes pipes, pump stations, and valves. Eventually, $w$ is a function that assigns a weight to each edge quantifying the physical characteristics (diameter, length, roughness) for each pipe \cite{Giudicianni2018}. In particular, WDNs are strongly constrained by their geographical embedding \cite{boccaletti2006complex} in that connections between distant nodes are
unlikely to be found, due to physical and economic constraints. 

\subsection{Benchmarking water distribution system}
MS-SP is firstly tested on the real medium-size Colorado Springs (US) water utility - which  currently serves a population of about 370,000 inhabitants. Figure~\ref{fig:Colorado1} shows its network layout. This encompasses 1,782 junctions and 4 reservoirs ($n =$ 1,786 nodes), 1,985 pipes, 6 pumps and 4 valves ($m = 1995$ links). Figure~\ref{fig:Colorado2} clearly demonstrates the size reduction of the Colorado water network after its transformation into a MS network. 
From a visual analytics point of view, Figure~\ref{fig:Colorado2} naturally highlights both the more inter-connected network areas and bottleneck links likely related to most vulnerable parts of the system.

\begin{figure}[!ht]
    \centering
    \begin{subfigure}[b]{0.9\linewidth} 
        \centering
        \includegraphics[width=\textwidth]{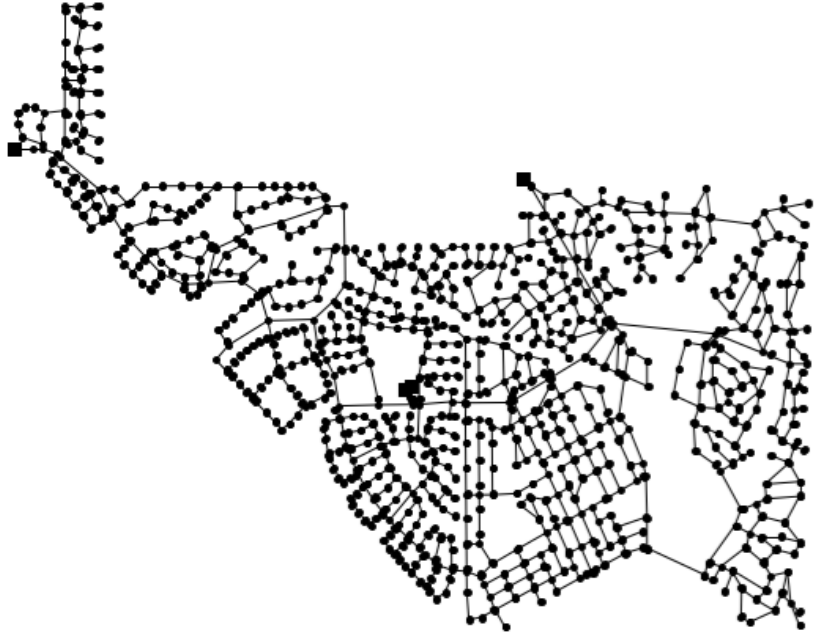}
        \caption{Water network layout of Colorado Springs}
        \label{fig:Colorado1}
    \end{subfigure}
    \\
    \begin{subfigure}[b]{0.9\linewidth}
        \centering
        \includegraphics[width=\textwidth]{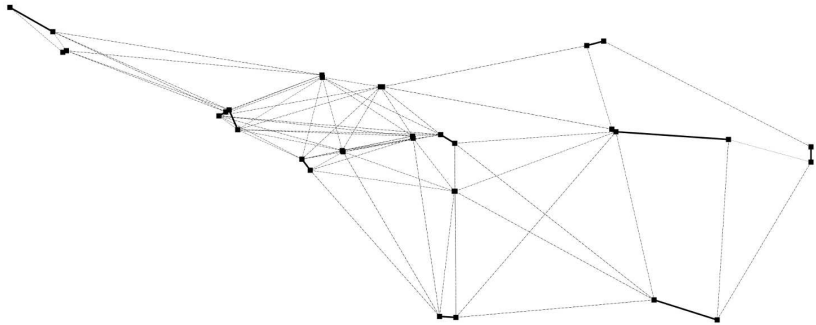}
        \caption{Multiscale water network of Colorado Springs}
        \label{fig:Colorado2}
    \end{subfigure}
    \caption{Multiscale dimensionality reduction for Colorado Springs water network}
    \label{fig:Colorado_} 
\end{figure}

\subsection{Large scale water distribution system}
The second case-study corresponds to the large-scale water utility which serves the Spanish city of Alcal\'a de Henares. It is located 22 miles northeast of the country's capital, Madrid, and it counts on a population of 201,000 inhabitants. The water distribution network model (see Figure~\ref{fig:Alcala1}) encompasses 11,473 junctions, 3 reservoirs ($n$ = 11,476 nodes), and 12,454 pipes, ($m$ = 12,454 links).  Figure~\ref{fig:Alcala2} shows the corresponding MS network layout. 

\begin{figure}[!ht]
    \centering
    \begin{subfigure}[b]{0.9\linewidth} 
        \centering
        \includegraphics[width=\textwidth]{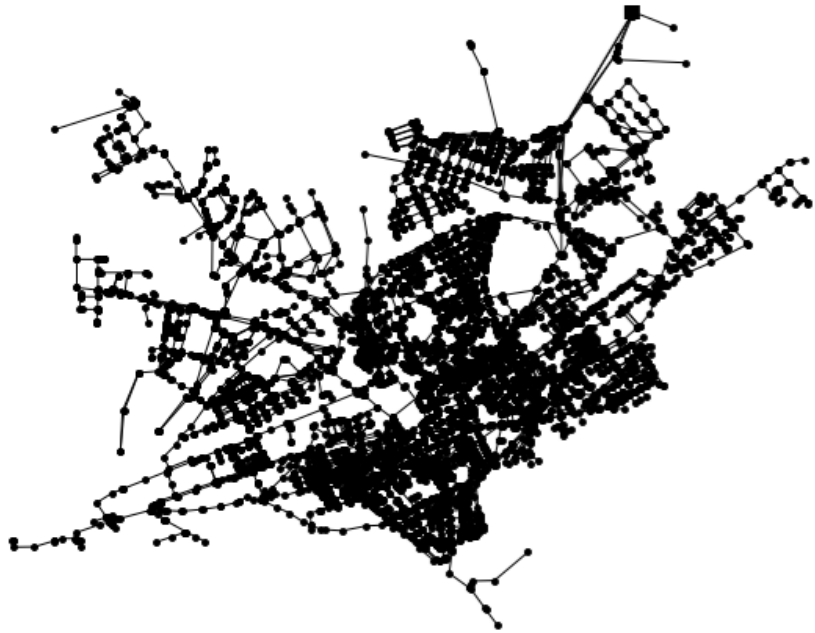}
        \caption{Water network layout of Alcal\'a}
        \label{fig:Alcala1}
    \end{subfigure}
    \\
    \begin{subfigure}[b]{0.9\linewidth}
        \centering
        \includegraphics[width=\textwidth]{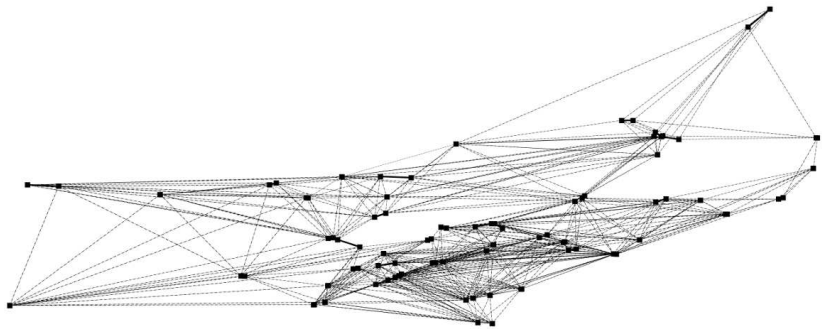}
        \caption{Multiscale water network of Alcal\'a}
        \label{fig:Alcala2}
    \end{subfigure}
    \caption{Multiscale dimensionality reduction for Alcal\'a water network}
    \label{fig:Alcala_} 
\end{figure}

\section{Results}
\label{sec:results}
The first results this section introduces are those corresponding to the topological analysis of the original water network and the MS network, for both case studies. A comparison between the original layout and the collapsed one is carried out in terms of: 

\begin{itemize}
    \item \textit{Links Density} $q$ which is the ratio between the total number $m$ of network edges and the maximum number of edges $m^*= n(n-1)/2$ of a network with $n$ nodes:
                                                           
\begin{equation}
q=\frac{2m}{n(n-1)}
\end{equation}  
                                                          
    \item \textit{Average Node Degree} $\overline{K}$ is the average value of the node degree $k_i$ (number of edges concurring in the node) over all nodes $n$: 

\begin{equation}
\overline{K}=\frac{2m}{n}
\end{equation}

    \item \textit{Diameter} $D$ \citep{Watts1998} is defined as the maximum shortest distance (the maximum geodesic length) $d_{ij}$ between any pair of vertices $i$ to node $j$ (computed as the number of edges along the shortest path connecting them): 

\begin{equation}
D=\max d_{ij}
\end{equation}

    \item \textit{Average Path Length} $l$ \citep{Watts1998} is the average number of steps along the shortest paths for all possible pairs of nodes in the network: 
\begin{equation}
l =\frac{2\sum d_{ij}}{n(n-1)}
\end{equation}

     \item \textit{Algebraic Connectivity} $\lambda_2$ \citep{Fiedler1973} corresponds to the second smallest eigenvalue of graph Laplacian matrix $L$ 
    
     \item \textit{Spectral Gap} $\Delta \lambda$ \citep{Estrada2006} is the difference between the first and second eigenvalue of the adjacency matrix $A$.
\end{itemize}

Table \ref{tab:topo_network} enumerates the main topological metrics computed for both case studies on the original and the MS network. The total number of links $m_b$ for the MS network is equal to the sum of the boundary links $m_{ex}$ and the internal hyper-links $m_{in}$. The size problem reduction is evident on nodes (from $n = 1,786$ to $n_b = 33$ for Colorado and from $n = 11,476$ to $n_b = 114$ for Alcal\'a) and also on links (from $m$ = 1992 to $m_b$ = 83 for Colorado and from $m = 12.454$ to $m_b = 596$ for Alcal\'a). The connectivity $\overline{K}$ strongly increases for the both the MS network (from $\overline{K} = 2.23$ to $\overline{K} = 5.03$ for Colorado and from $\overline{K} = 2.17$ to $\overline{K} = 10.46$ for Alcal\'a).

The dimensionality reduction working with the MS network makes the network density increases up to 2 orders of magnitude (from $q$ = 0.001 to $q$ = 0.157 for Colorado and from $q$ = 0.0002 to $q$ = 0.0933 for Alcal\'a). This augmented inter-connectivity is also reflected by the two spectral metrics measuring the robustness of network. The algebraic connectivity and the spectral gap also increase when moving from the original to the MS network, as it is shown in Table~\ref{tab:topo_network}. These changes in the topological metrics reflect the shift in the structure of the MS network which can be regarded now as low interconnected small-world clusters (whose links are the internal hyper-links). In fact, after the size reduction due to the application of the MS-SP algorithm, each cluster of the MS network becomes into a fully connected layout, weakly linked to other clusters by the boundary links. The typical small-world behaviour is also confirmed by the low value of the communication metrics as they are the diameter and the average path length which scale approximately with the $\log(n)$, as happens for the small world networks (Table~\ref{tab:topo_network}).

\begin{table}[!ht]\scriptsize
\caption{Topological characteristics of the original water network and the MS network layout for Colorado Springs and Alcal\'a de Henares} 
\label{tab:topo_network} 
\centering 
\scalebox{1.0}{
\begin{tabular}{c c c c c c}
\hline
\textbf{Metric}  & \textbf{Colorado} & \textbf{Colorado-MS} & \textbf{Alcal\'a} & \textbf{Alcal\'a-MS}\\
\hline
  $n$ or $n_b$  & 1786  & 33 &  11,476 & 114 \\
  $m$ or $m_b$  & 1995  & 83 &  12,454 & 596 \\
  $\overline{K}$  & 2.23  & 5.03 & 2.17 & 10.46 \\
  $q$  & 0.0012  & 0.1571 & 0.0002 & 0.0933 \\
  $D$  & 69  & 8 & 163 & 9 \\
  $l$  & 25.94  & 3.15 & 64.88 & 3.87 \\
  $\lambda_2$  & 0.00053  & 0.23512 & 0.00009 & 0.15884 \\
  $\Delta \lambda$ & 0.1293  & 0.1735  & 0.0957 & 0.0587 \\
\hline  
\end{tabular}}
\end{table}

The simulation results for Colorado and Alcal\'a water utilities are reported in Table~\ref{tab:Colorado} and Table~\ref{tab:Alcala}, respectively. A suitable number of clusters $C$ is taken on both cases to optimise the overall connectivity of the partitioned network, according to the relationship $C_{opt} \propto n^{0.28}$ reported in \cite{Giudicianni2018}, where $C_{opt}$ is the optimal number of clusters from a topological point of view. As a result, the number of clusters for Colorado is set to $C=8$, while $C=13$ for Alcal\'a's network. Up to 10 paths are generated by connecting random pairs of source and target to validate the proposed MS-SP algorithm. For each pair, the shortest path is computed by running the code 10 times and averaging the computational time.

\begin{table}[!ht] \scriptsize
\caption{Simulation results for the Colorado Springs water network} 
\label{tab:Colorado} 
\resizebox{0.95\columnwidth}{!}{%
\centering 
\begin{tabular}[t]{c c c c c c}
\hline
\textbf{Pairs} & \textbf{D-SP value} & \textbf{MS-SP value} & \textbf{D-SP time} & \textbf{MS-SP time} & \textbf{Red. time} \\ 
 & [-] & [-] & [s] & [s] & [\%] \\
\hline
  1  & 13  & 13 & 0.0010 & 0.0001 & 90.0 \\
  2  & 21  & 21 & 0.0015 & 0.0005 & 66.6 \\
  3  & 29  & 29 & 0.0022 & 0.0006 & 72.6 \\
  4  & 33  & 33 & 0.0026 & 0.0007 & 72.9 \\
  5  & 38  & 38 & 0.0032 & 0.0004 & 87.4 \\
  6  & 41  & 41 & 0.0033 & 0.0006 & 81.7 \\
  7  & 52  & 52 & 0.0043 & 0.0007 & 83.6 \\
  8  & 56  & 56 & 0.0036 & 0.0005 & 85.8 \\
  9  & 60  & 60 & 0.0041 & 0.0006 & 85.2 \\
  10 & 66  & 66 & 0.0039 & 0.0004 & 89.6 \\ 
\hline
\end{tabular}
}
\end{table}

MS-SP algorithm provides the exact value of the shortest path between each pairs of randomly generated source and target nodes. This represents a clear advantage with respect to previous methodologies whom provide approximated results. Table \ref{tab:Colorado} and Table \ref{tab:Alcala} clearly state the D-SP and the MS-SP provide the same results (difference is equal to zero).

\begin{table}[!ht] \scriptsize
\caption{Simulation results for the Alcal\'a water network} 
\label{tab:Alcala} 
\centering 
\resizebox{0.95\columnwidth}{!}{%
\begin{tabular}{c c c c c c}
\hline
\textbf{Pairs} & \textbf{D-SP value} & \textbf{MS-SP value} & \textbf{D-SP time} & \textbf{MS-SP time} & \textbf{Red. time} \\ 
 & [-] & [-] & [s] & [s] & [\%] \\
\hline
  1  & 32  & 32 & 0.0007 & 0.0003 & 50.2 \\
  2  & 40  & 40 & 0.0019 & 0.0004 & 80.5 \\
  3  & 53  & 53 & 0.0032 & 0.0005 & 84.7 \\
  4  & 60  & 60 & 0.0041 & 0.0006 & 85.1 \\
  5  & 72  & 72 & 0.0027 & 0.0004 & 84.4 \\
  6  & 88  & 88 & 0.0098 & 0.0010 & 90.1 \\
  7  & 94  & 94 & 0.0102 & 0.0015 & 85.3 \\
  8  & 102 & 102 & 0.0129 & 0.0011 & 91.6 \\
  9  & 115 & 115 & 0.0094 & 0.0015 & 82.5 \\
  10 & 116 & 116 & 0.0097 & 0.0017 & 91.9 \\ 
\hline  
\end{tabular}}
\end{table}

The computational time for D-SP algorithm grows with the distance between source and target nodes. However, computational times for MS-SP show to be a plateau value of an order of magnitude smaller than that D-SP method. This is clearly shown in figures \ref{fig:ColoradoT} and \ref{fig:AlcalaT} (with the results on Colorado and Alcal\'a utility networks).

\begin{figure}[!ht]
\centering
\includegraphics[width=0.9\linewidth]{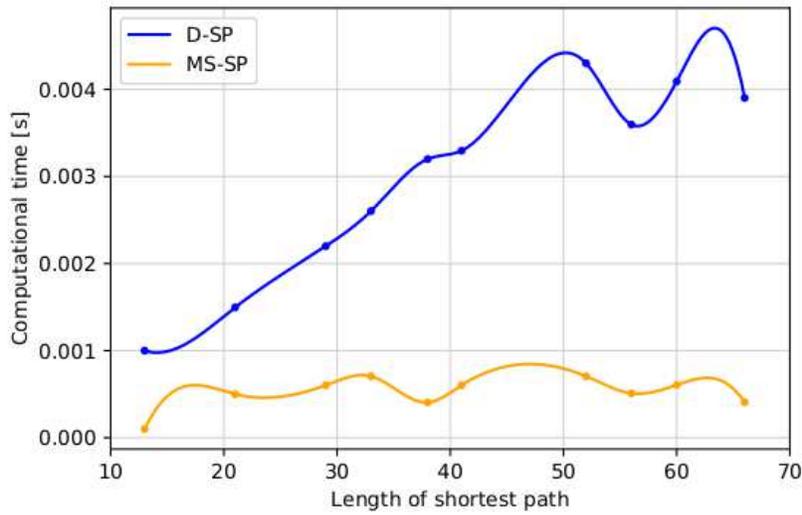}
\caption{Computational time for D-SP and MS-SP algorithms, for Colorado water network}
\label{fig:ColoradoT}
\end{figure}

\begin{figure}[!ht]
\centering
\includegraphics[width=0.9\linewidth]{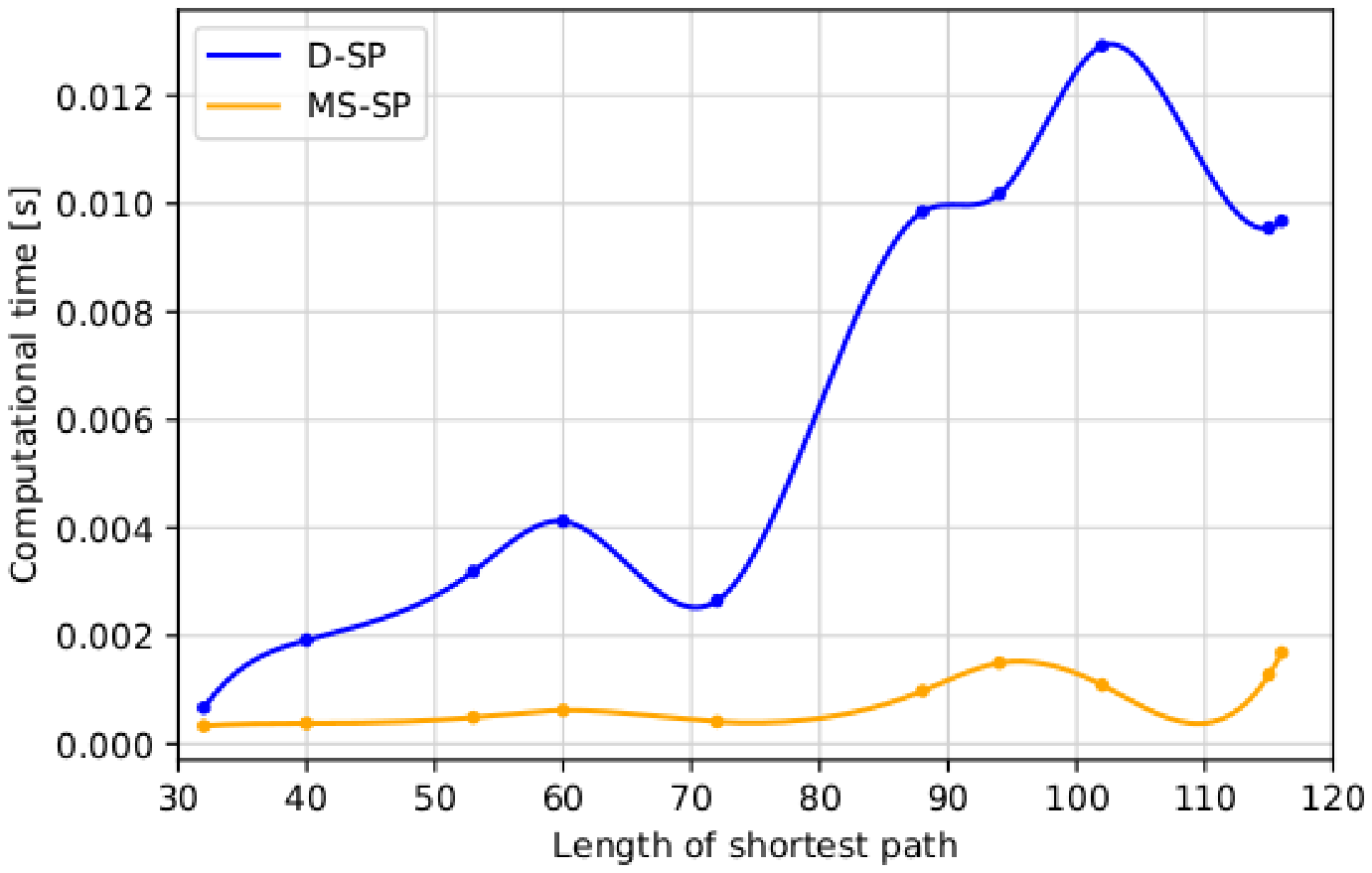}
\caption{Computational time for D-SP and MS-SP algorithms, for Alcal\'a water network}
\label{fig:AlcalaT}
\end{figure}

Table \ref{tab:Colorado} shows the difference in percentage between the D-SP and the proposed MS-SP computational time for Colorado. This difference on time ranges from 66\% to 90\%. Table \ref{tab:Alcala} shows the difference in percentage between the D-SP and the proposed MS-SP computational time for Alcal\'a. This difference on time ranges from 50\% to 92\%. Both differences on computational time stand as a conspicuous time reduction for computing the shortest path.

The MS-SP algorithm is implemented in Python 3.6. All the simulations were run on a Linux Xubuntu 16.04 PC with 2.13~GHz Intel\textregistered $\;$ Core\texttrademark $\;$ i3 CPU m330 64 GB of memory and 4.00 GB of RAM.

\section{Discussion}
\label{sec:discussion}
An appealing by-product of the current research is the Small-World properties of the collapsed layout of the original network which preserves essential information about the system, but it is simultaneously characterised by a dramatically size reduction. Indeed, one of the most useful management strategies for WDNs is the sectorisation, which consists in dividing the system in smaller and monitored districts, connected each other by a few number of pipes. Sectorisation helps the pressure management, the leakage detection and reduction, the spreading reduction of contaminants. For this reason, clustering algorithms have been largely applied in hydraulic engineering to define the optimal configuration of districts that balance the aforementioned positive aspects to the resilience reduction that could happen as consequence of the closure of some boundary pipes (boundary links $m_{ex}$ in the model). The spatial-temporal variability of the functioning conditions, such as the change in water request all over the system, could invalidate the designed sectorisation, by compromising the performance of the WDN. An adaptive/dynamic approach could be a valid solution, providing the aggregation/desegregation of districts (clusters) according to specific conditions. In this regard, the following constraints should be respected: 
\begin{itemize}
        \item the new aggregate districts include the former ones without splitting them;
        \item the previous clustering layout and the devices already installed have to be exploited; 
        \item the set of new boundary links is included in the set of boundary links of the original partitioning; 
\end{itemize}

As a consequence of working with the MS network layout the management of the system is simplified. In addition, the new investment costs is minimised, even nullified, and the computational burden of the whole procedure is reduced. This structural knowledge comes in the form of pairwise must-link (boundary links) and cannot link (internal links) constraints to be respected at each step of clustering by means a semi-supervised approach (which is particularly suitable for working with real-world systems if background knowledge about the structure are available).

The MS network implicitly takes into account the aforementioned structural knowledge, thanks to the shift in the structure to a low interconnected small-world clusters. In this way, it is ensured that any clustering algorithm provides a solution in which the novel set of boundary links constitutes a sub-set of the boundary links of the original cluster layout. On top of this, a network community detection algorithm splits a network in such way that each cluster is formed by elements having a high density connection between each other and a lower probability to be connected to items belonging to other clusters. 

\begin{figure}[!ht]
    \centering
    \begin{subfigure}[t]{0.45\textwidth}
        \includegraphics[width=\textwidth]{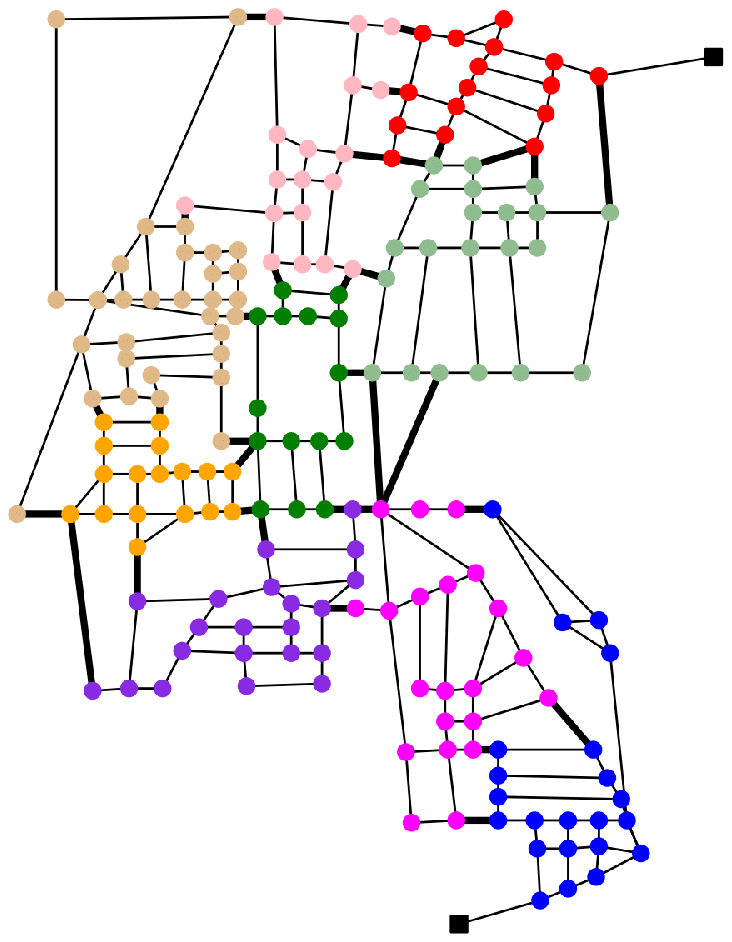}
        \caption{Actual district layout for a WDN}
        \label{fig:3a}
    \end{subfigure}
    ~
    \begin{subfigure}[t]{0.45\textwidth}
        \includegraphics[width=\textwidth]{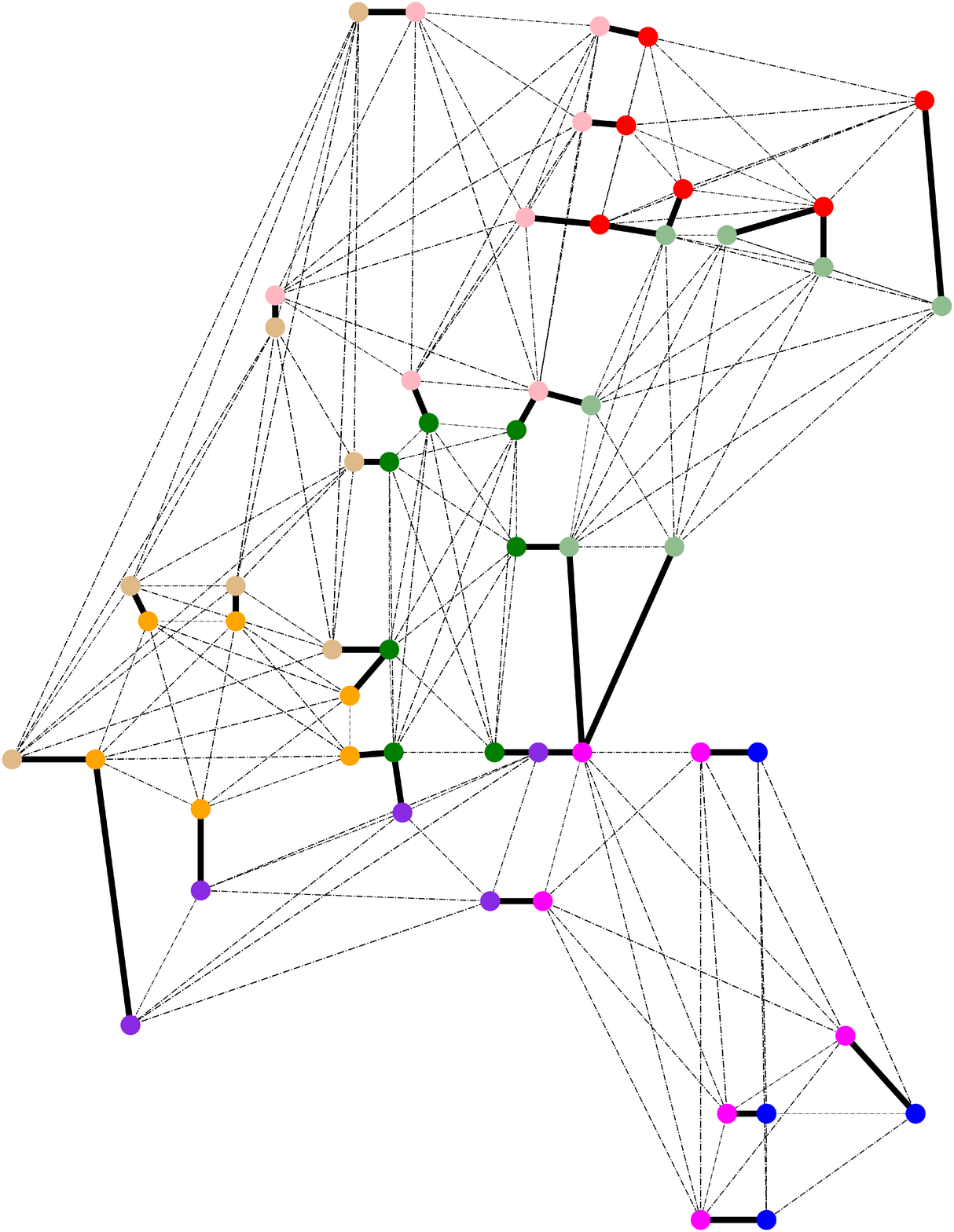}
        \caption{MS layout for actual district layout}
        \label{fig:3b}
    \end{subfigure}
    ~
    \begin{subfigure}[t]{0.45\textwidth}
        \includegraphics[width=\textwidth]{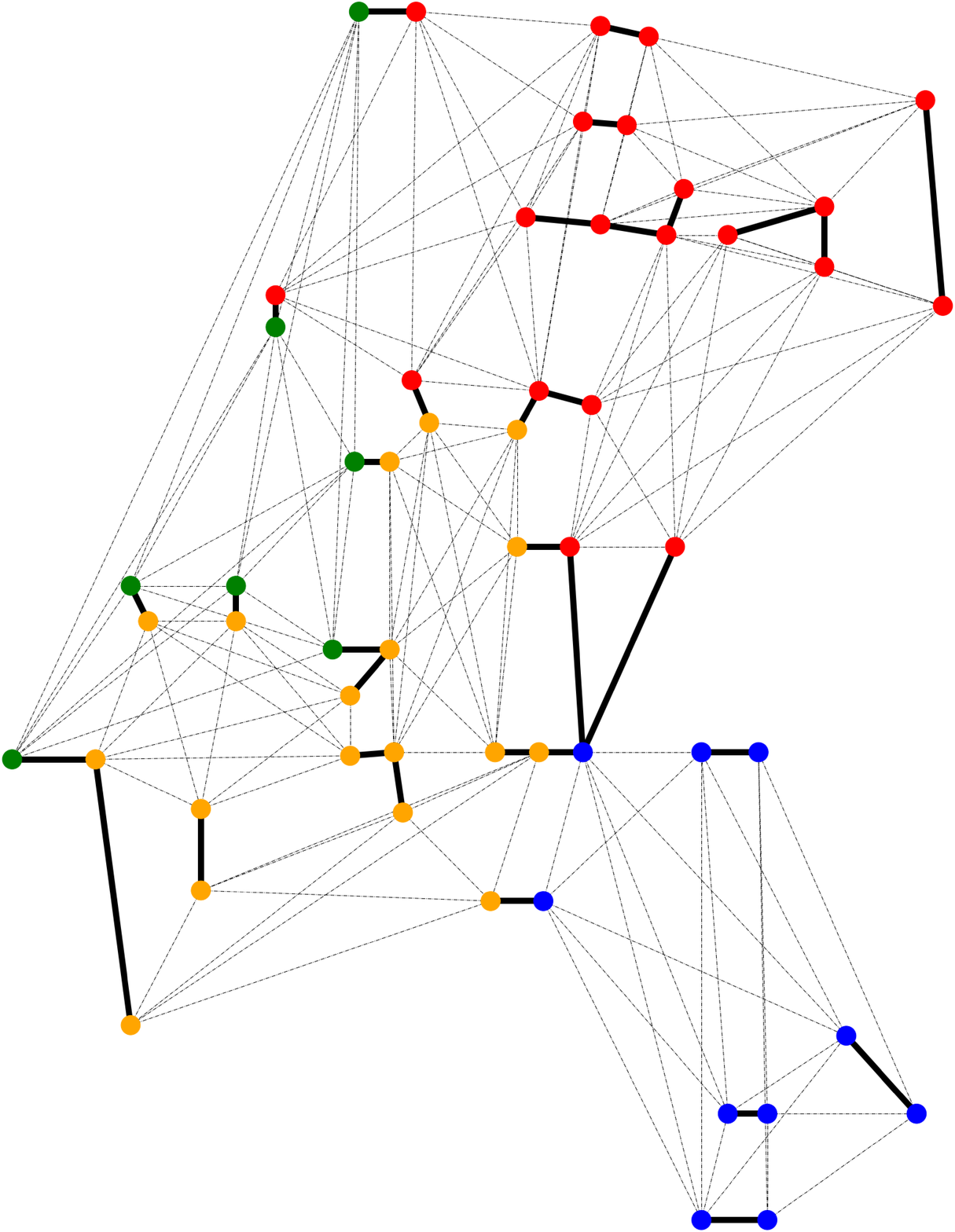}
        \caption{MS layout for new district layout}
        \label{fig:4a}
    \end{subfigure}
    ~
    \begin{subfigure}[t]{0.45\textwidth}
        \includegraphics[width=\textwidth]{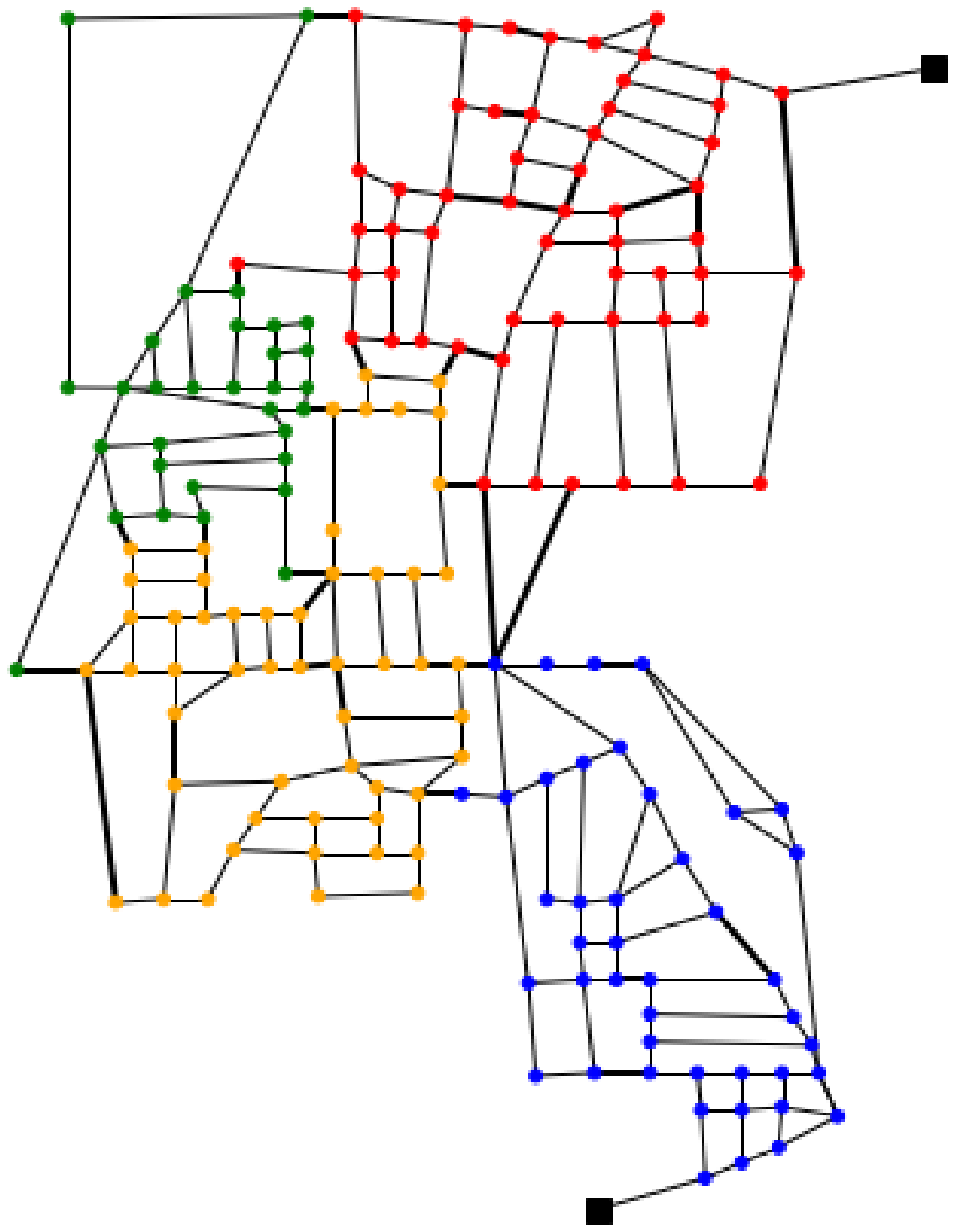}
        \caption{New district layout for a WDN}
        \label{fig:4b}
    \end{subfigure}
    \caption{Graphical explanation for the dynamic aggregation / disaggregation process for the partitioning of a WDN through the application of the MS layout}
    \label{fig:MSprocess} 
\end{figure}

\textcolor{black}{The steps of the procedure are shown in Figure \ref{fig:MSprocess}. First, the actual clustering layout of the WDN is detected (see Figure \ref{fig:3a}) and then the corresponding MS network is built (see Figure \ref{fig:3b} in which it is evident the size reduction of the WDN) This figure also shows the key elements, such as the boundary nodes of each cluster (highlighted by their corresponding district colour at Figure Figure \ref{fig:3a}), the boundary links (bold black line), and the internal links (thin dashed grey line). The set of hyper-links (boundary links and internal links) is crucial for automatically implementing the semi-supervised clustering during the aggregation phase. Boundary links represent the connectivity between different clusters, while the internal links constitute the internal connectivity of each cluster (according to the shortest path linking each pairs of boundary nodes belonging to the same cluster). After that, the previous districts are aggregated in the MS network (see Figure \ref{fig:4a}), by applying a clustering algorithm to provide the new clustering layout. Finally, in Figure \ref{fig:4b} the new district configuration is shown, with the new bigger clusters that perfectly include the former ones (avoiding to split them), and the boundary set which constitutes a sub-set of the previous one. This dynamic aggregation / disaggregation process ensures that the districts in each phase are kept in control, exploiting the devices already installed in the WDN.}

\section{Conclusions}
\label{sec:conclusions}
This paper proposes a novel method to efficiently solve the shortest path problem for large-scale water networks (and other utility networks), showing the potentialities on its application to their management and protection. The algorithm is based on a community structure principle, which aids to collapse the original network into a limited set of key elements. This is made through the novel concept of a multiscale (MS) network that reduces the original system into a series of interconnected, landmark, nodes. These landmark nodes are connected by a family of links coming both from the original network and the scale process where links and nodes are aggregated into hyper-links. The hyper-links are weighted by the shortest-path distances between their corresponding extreme nodes. The MS network structure easier the shortest path computation as the network's dimensionality is significantly reduced. Computing the shortest path is broken down to be done in several but smaller areas instead of directly using the whole original network. 
The simulation results over two urban water utilities confirm the efficiency of the proposed multiscale shortest path algorithm. This provides an exact solution for the problem in a significantly lower computational time than using Dijkstra's algorithm. The approach conveniently scales to big-size networks.
\textcolor{black}{The reduced size of the multiscale network coming from the application of the proposed algorithm results in a dramatic advantage for speeding up and simplifying the management of water systems by means the definition of optimal metered districts.}

Further works will benefit of this process since, in general, links connecting key nodes are often closed (i.e. most of the boundary pipes connecting nodes belonging to different network areas). In addition to utility networks applications, the multiscale shortest path usefulness extends to critical infrastructures such as transportation, communication, logistic and supply networks. The proposal will also improve the communication speed in general big-size networked systems and will back-up further near real-time operations. 


\bibliographystyle{model2-names}
\bibliography{refs}

\end{document}